\documentclass[11pt]{article}
\usepackage[a4paper,margin=1.0in]{geometry}

\usepackage{thm-restate}									% BUG: loading thm-restate/thmtools
\usepackage{amsthm}											% before amsthm creates a warning
\usepackage{amssymb}										% but the other order creates an error
\usepackage{amsmath}
\usepackage{url}
\usepackage{enumerate}
\usepackage[dvipsnames,svgnames,table]{xcolor}
\usepackage[algoruled,vlined,linesnumbered]{algorithm2e}
\usepackage{authblk}										% for nicer affiliation blocks
\usepackage{setspace}										
\usepackage[colorlinks=true, urlcolor=blue, linkcolor=blue, citecolor=ForestGreen]{hyperref}
\usepackage[noabbrev, nameinlink]{cleveref}					% Must be included after ›hyperref‹.
\usepackage{tikz}
\usetikzlibrary{arrows.meta, positioning}
\usepackage{pgfplots}
\pgfplotsset{compat=1.18}

\newcommand{\ignore}[1]{{}}

\SetNlSty{bfseries}{\color{black}}{}

\newtheorem{theorem}{Theorem}
\newtheorem{definition}[theorem]{Definition}
\newtheorem{lemma}[theorem]{Lemma}

\newcommand*{\Otilde}{\widetilde{O}}
\newcommand*{\G}{\mathcal{G}}

\newcommand*{\nwspace}{\hspace*{.1em}} % MS: a space for math environments which is smaller than '\,'

\newcommand*{\poly}{\textsf{poly}}

\let\oldsqrt\sqrt
\def\hksqrt{\mathpalette\DHLhksqrt}
\def\DHLhksqrt#1#2{\setbox0=\hbox{$#1\oldsqrt{#2\,}$}\dimen0=\ht0
   \advance\dimen0-0.2\ht0
   \setbox2=\hbox{\vrule height\ht0 depth -\dimen0}%
   {\box0\lower0.4pt\box2}}
\renewcommand\sqrt\hksqrt

\renewcommand{\leq}{\leqslant}

\renewcommand{\le}{\leqslant}
\renewcommand{\ge}{\geqslant}

\title{Near-Optimal Replacement Path Coverings}

\author[1]{Davide Bilò}
\author[2]{Keerti Choudhary}
\author[3]{Sarel Cohen}
\author[4]{Martin Schirneck}

\affil[1]{Department of Information Engineering, Computer Science

	and Mathematics, 	
	University of L'Aquila

	\texttt{davide.bilo@univaq.it}
	\vspace*{.5em}
}
\affil[2]{Department of Computer Science and Engineering, 

	Indian Institute of Technology Delhi,

	\texttt{keerti@iitd.ac.in}
	\vspace*{.5em}
}
\affil[3]{Efi Arazi School of Computer Science, 

	Reichman University

	\texttt{sarel.cohen@runi.ac.il}
	\vspace*{.5em}
}
\affil[4]{Department of Informatics, 

	Karlsruhe Institute of Technology

	\texttt{martin.schirneck@kit.edu}
}

\date{}

\begin{document}
\maketitle

\begin{abstract}
\noindent
	Let $L$ and $f$ be positive integers.
	An $(L,f)$\emph{-replacement path covering} (RPC) for a graph $G$ is a family 
	$\mathcal{G}$ of subgraphs such that,
	for every set $F$ of at most $f$ edges,
	there is a subfamily $\mathcal{G}_F \,{\subseteq}\, \mathcal{G}$ with the following properties.
	\begin{enumerate}
		\item No subgraph in $\mathcal{G}_F$ contains an edge of $F$.
		\item For each pair of vertices $s,t$ that have a shortest path in $G{-}F$ 
		with at most $L$ edges,\\ one such path also exists in some subgraph in $\mathcal{G}_F$.
	\end{enumerate}
	The  total number $|\mathcal{G}|$ of subgraphs is called the \emph{covering value}.
	
	RPCs are important tools in the design of fault-tolerant data structures.
	Weimann and Yuster [TALG 2013] presented an RPC with covering value
	$\Otilde(f L^f)$.
	Karthik and Parter [TALG 2024] showed that $\Omega( (L/f)^f )$ subgraphs are necessary.
	Recently, Bilò, Chechik, Choudhary, Cohen, and Schirneck [ICALP 2026]
	devised a new approach for very small sensitivities $f = o(\log L)$ with covering value $\Otilde(f e^f (L/f)^{f+o(1)})$.
	They also showed that any RPC in the complementary range $f = \Omega(\log L)$
	must contain $\Omega( (\sqrt{f e^f}/L) \cdot (L/f)^f)$ subgraphs.
	The asymptotically optimal covering value has remained open.
	
	We present two surprisingly simple constructions, improving both previous upper and lower bounds. Our upper bound follows from a single observation: sampling each edge with the optimal probability $f/(L+f)$, rather than the classical $1/L$, already yields a nearly optimal covering value. Together with a matching lower bound, this establishes a near-tight covering value of 
	$\widetilde{\Theta}\big(\frac{(L+f)^{L+f}}{L^L  f^f} \big) \cdot \poly(f)$ for the much wider range of $f = O(L)$.
\end{abstract}

% plain LaTeX abstract
\ignore{%
	Let $L$ and $f$ be positive integers. An $(L,f)$-replacement path covering (RPC) for a graph $G$ is a family $\mathcal{G}$ of subgraphs such that, for every set $F$ of at most $f$ edges,
there is a subfamily $\mathcal{G}_F \subseteq \mathcal{G}$ with the following properties. (1) No subgraph in $\mathcal{G}_F$ contains an edge of $F$. (2) For each pair of vertices $s,t$ that have a shortest path in $G{-}F$	with at most $L$ edges, one such path also exists in some subgraph in $\mathcal{G}_F$. The  total number $|\mathcal{G}|$ of subgraphs is called the covering value.
	
RPCs are an important tools in the design of fault-tolerant data structures. Weimann and Yuster [TALG 2013] presented an RPC with covering value $\widetilde{O}(f L^f)$. Karthik and Parter [TALG 2024] showed that $\Omega( (L/f)^f )$ subgraphs are necessary. Recently, Bilò, Chechik, Choudhary, Cohen, and Schirneck [ICALP 2026] devised a new approach for very small sensitivities $f = o(\log L)$ with covering value $\widetilde{O}(f e^f (L/f)^{f+o(1)})$. They also showed that any RPC in the complementary range $f = \Omega(\log L)$ must contain $\Omega( (\sqrt{f e^f}/L) \cdot (L/f)^f)$ subgraphs. This left open the question of what is the true covering value.
	
We give two surprisingly simple constructions that improve both the upper and lower bound.	This results in a near-tight covering value of $\widetilde{\Theta}(\frac{(L+f)^{L+f}}{L^L  f^f}) \cdot \mathsf{poly}(f)$ for the much wider range of $f = O(L)$.
}

\section{Introduction}
\label{sec:intro}

Fault-tolerant graph data structures are able to quickly report properties
of the underlying network, like pairwise distances or connectivity,
even when the input undergoes a bounded number of transient edge failures.
Research in fault tolerance (a.k.a.\ sensitivity analysis) bridges the gap between the extensive body of foundational work on static graph algorithms,
where the input remains fixed during the whole computation,
and the requirements of real-world scenarios,
where the network is constantly changing.
In the last two decades, fault-tolerant data structures have been developed, for example, for connectivity problems~\cite{BrSa19, DuanP09a, DuanP10, DuanP:17, Petruschka26ConnectivityLabelingFaultyColoredGraphs, PatrascuT:07},
flows and cuts~\cite{Ahi26MaxFlowMinCutSensitivityOracles,Baswana23stCutsSensitivityOracle},
shortest paths~\cite{ChCo20, ChCoFiKa17,  ChoS024, DeThChRa08, DeyGupta24, DuanP09a, DuRe22, GuRen21,GW12}, diameter and eccentricity~\cite{Bilo22Extremal,Bilo23STDiameter,HenzingerL0W17}, or routing~\cite{CLPR12}.
They have also found applications for \textsf{NP}-hard problems like vertex cover, $k$-path, and $k$-clique~\cite{AlmanHirsch22ExteriorAlgebras,Bilo22FixedParameterSensitivityOracles}.

The prototypical fault-tolerant data structure is a \emph{distance sensitivity oracle} (DSO).
It preprocesses a given (potentially directed or weighted) graph $G = (V,E)$ and a positive integer $f$,
called the \emph{sensitivity}.
Every query for the oracle is a triplet $(s,t,F)$ consisting of two vertices $s,t \in V$  and a set $F \subseteq E$ of at most $f$ edges.
The output of the oracle is the \emph{replacement distance} $d_{G{-}F}(s,t)$,
the length of a shortest path from $s$ to $t$ in $G{-}F$.
Any such shortest path is called a \emph{replacement path}.
When designing such a DSO, a typical approach is to set some cut-off value $L$ 
(determined by the particular approach)
and treat replacement paths with at most $L$ edges
separately, see e.g.\ \cite{Bilo24ApproxDSOSubquadraticTheoretiCS,BrSa19,GW12,WY13}.
We call this the \emph{hop-short case}.
For the complete oracle, one still needs a separate mechanism to combine the hop-short paths
to an answer for the general case.
We only consider the hop-short case in this work.

The main tool for solving this problem is \emph{replacement path coverings} (RPCs)\footnote{%
	The name was introduced in the work by 
	Karthik and Parter~\cite{KarthikParter24DeterministicRPC_TALG}.
}
as introduced by Weimann and Yuster~\cite{WY13}.
Their preprocessing algorithm generates a family $\G$ of subgraphs of $G$.
It starts with $|\G| = \Otilde(f L^f)$ copies\footnote{
	The $\Otilde$-notation hides a $\log(n)$-factor.
} 
of the original graph and then,
in each one of them, removes any edge independently with probability $1/L$.
The authors of \cite{WY13} show that, w.h.p.\footnote{%
	By \emph{high probability} we mean $1-n^{-c}$ for some constant $c > 0$.
	In all cases mentioned here, $c$ can be chosen arbitrarily large without 
	affecting the asymptotics.
}\
for all queries $(s,t,F)$
such that $s$ and $t$ are joined in $G{-}F$ by some shortest path on at most $L$ edges,
there exists a subgraph $G^* \in \G$ that does not have any edge of $F$
but at least one of those shortest paths survives in $G^*$.
This is useful for DSOs since the properties of an RPC imply that the $s$-$t$-distance in $G^*$
is precisely the replacement distance $d_{G{-}F}(s,t)$.
%RPCs later found applications in fault tolerance also beyond distances~\cite{Bilo22FixedParameterSensitivityOracles}.

The remaining question is how to identify such a subgraph $G^*$.
The approach in~\cite{WY13} is to scan the whole family $\G$ and filter for those subgraphs
that do not have any edge of $F$.
This takes time $O(f)$ per subgraph and
results in a subfamily $\G_F \subseteq \G$.
Then, $\min_{H \in \G_F} d_{H}(s,t)$ is computed to obtain the replacement distance.
Subsequent works~\cite{Bilo24ApproxDSOSubquadraticTheoretiCS,BiChChCoFrScFOCS24,KarthikParter24DeterministicRPC_TALG} have also focused more on
providing a small set $\G_F$ of suitable subgraphs for a given failure set $F$
instead of finding exactly $G^*$.
This leads to the following definition.

\begin{definition}[replacement path coverings]
\label{def:RPC}
	Let $L$ and $f$ be positive integers and $G = (V, E)$ a graph.
	An $(L,f)$\emph{-replacement path covering} for $G$ is a family $\G$ of spanning subgraphs of $G$
	that has a subfamily $\G_F \,{\subseteq}\, \G$ for every set $F \,{\subseteq}\, E$ of $|F| \,{\leqslant}\, f$ edges
	such that the following two properties hold.
	%\vspace*{.25em}
 	\begin{enumerate}
 		\item No subgraph in $\G_F$ contains an edge of $F$.
 	%	\vspace*{.25em}
 		\item For all vertices $s,t \in V$ for which there exists a shortest path from $s$ to $t$ in $G{-}F$  
 			with at most $L$ edges, 
 			at least one subgraph in $\G_F$ also has such a path.
 	\end{enumerate}
\end{definition}

As usual with data structures, there are different parameters to optimize and 
there is not a single solution that is best in all categories.
Possible optimization criteria are the time to preprocess $\G$,
the time to compute $\G_F$ given a query $(s,t,F)$,
the total number of subgraphs $|\G|$, or the size of the subfamily $\G_F$.
We are interested in minimizing the cardinality $|\G|$,
known as the \emph{covering value}~\cite{KarthikParter24DeterministicRPC_TALG}.
In fault-tolerant data structures, 
the covering value translates to the space required by the oracle.
The $(L,f)$-replacement path covering in the original paper by Weimann and Yuster~\cite{WY13} has covering value $\Otilde(f L^f)$.
Karthik and Parter~\cite{KarthikParter24DeterministicRPC_TALG} derandomized
the construction using heavy algebraic machinery including error-correcting codes.
This, however, increased the covering value to $O((Cf L \log n)^{f+1})$ 
for some constant $C > 0$.

They also showed that $\Omega((L/f)^f)$ subgraphs are necessary.
Bilò, Chechik, Choudhary, Cohen, Friedrich, and Schirneck~\cite{Bilo25IndexingSubnetworks,Bilo26SimplerRPCs}
subsequently improved the upper and lower bounds for certain parameter ranges.
For very small sensitivities $f = o(\log L)$,
they gave a randomized construction with covering value $\Otilde(f e^f  \cdot (L/f)^{f+o(1)})$
via hierarchical sampling.
This was the first improvement over~\cite{WY13}.
They also showed a lower bound of $\Omega( (\sqrt{f e^f}/L) \cdot (L/f)^f)$,
which improves over the result in~\cite{KarthikParter24DeterministicRPC_TALG}
in the complementary range of $f = \Omega(\log L)$.
An overview of previous works can be found in \Cref{table:upper,table:lower}.

\begin{table}[t]
\caption{
Upper bounds on the covering value of $(L,f)$-replacement path coverings.
}
%\vspace*{.5em}
\centering
\setlength{\tabcolsep}{5pt}
\renewcommand{\arraystretch}{1.25}
\begin{tabular}{@{}ccc@{}}

\textbf{Covering Value} &  \textbf{Range} & \textbf{Reference} \\
\noalign{\hrule height 1pt}\\[-10pt]

$\Otilde(fL^f)$ &  no restrictions & \cite{WY13}\\[.25em]

$\Otilde\big((fL)^{f+1}\big)$ &  & \cite{KarthikParter24DeterministicRPC_TALG}\\[.25em]

\noalign{\hrule height .5pt}\\[-14pt]

$\Otilde\big(f \cdot L^{f + o(1)}\big)$ & $f = o(\log L)$ & \cite{Bilo25IndexingSubnetworks}\\[.25em]

$\Otilde\left(f \cdot  e^f \!\left(\frac{L}{f} \right)^{f+o(1)}\right)$ & & \cite{Bilo26SimplerRPCs}\\[.75em]

\noalign{\hrule height .5pt}\\[-14pt]

$\Otilde\left(f \cdot \frac{(L+f)^{L+f}}{f^f L^L}\right)$ & no restrictions & \Cref{thm:upper}
 
\end{tabular}
\label{table:upper}
\end{table}

While these results show that the previously known bounds are slack,
they leave the question of what is the \emph{true} covering value of an $(L,f)$-replacement path covering.
In this work, we show that it is $\widetilde{\Theta}\big(\frac{(L+f)^{L+f}}{L^L  f^f} \big) \cdot \poly(f)$,
even for large sensitivities $f\leq L$.
Our new upper bound, as well as the lower bound stem from simple constructions
and their proofs use only elementary means.
We exploit the sampling argument used in Lemma 10 of Bilò, Casel, Choudhary, Cohen, Friedrich, Lagodzinski, and Schirneck~\cite{Bilo22FixedParameterSensitivityOracles} in a different context. We observe that applying it to replacement path coverings, together with the optimal sampling probability $f/(L+f)$ instead of $1/L$ as in the construction of Weimann and Yuster~\cite{WY13}, yields the following theorem.
%The proof of Lemma 10 in Bilò, Casel, Choudhary, Cohen, Friedrich, Lagodzinski, and Schirneck~\cite{Bilo22FixedParameterSensitivityOracles} already contains the underlying sampling argument in the context of the $k$-path problem. Applying the same argument to replacement path coverings, while sampling each edge independently with probability $f/(f+L)$ instead of $1/L$ as in the construction of Weimann and Yuster~\cite{WY13}, immediately yields the following theorem.

%It is implicit in the work of Bilò, Casel, Choudhary, Cohen, Friedrich, Lagodzinski, and Schirneck~\cite{Bilo22FixedParameterSensitivityOracles} that the number of subgraphs can be decreased by taking the same approach as Weimann and Yuster~\cite{WY13}, but removing any edge in the copies with probability $f/(f+L)$ instead of just $1/L$. This immediately yields the following theorem.

\begin{restatable}{theorem}{upperbound}
\label{thm:upper}
Let $f, L$ be positive integers.
There exists a randomized algorithm that w.h.p.\ constructs an
$(L,f)$-replacement path covering with covering value 
$O\big(f \log(n) \cdot \frac{(L+f)^{L+f}}{f^f L^L} \big)$.
\end{restatable}

\noindent
The $O(f \log n)$-factor ensures a high success probability.
If it is enough that the algorithm succeeds in expectation, then $O\big(\frac{(L+f)^{L+f}}{f^f L^L} \big)$ subgraphs suffice.

We complement the upper bound with a nearly matching lower bound based on a simple layered graph construction for which every $(L,f)$-replacement path covering is large. 
In particular, we obtain the following result.

\begin{restatable}{theorem}{lowerbound}
\label{thm:lower}
Let $f, L$ be positive integers such that $L \ge 2$ and $f = O(L)$.
There exists a weighted directed graph for which any $(L,f)$-RPC
must have covering value $\Omega\big(\frac{1}{\sqrt f}  \cdot \frac{(L+f)^{L+f}}{L^L  f^f} \big)$.
\end{restatable}

Our results determine the covering value up to an $O(f^{3/2}\log n)$-factor.
As mentioned above, one $\Otilde(f)$-factor stems from probability boosting.
We expect that it can be removed in a careful derandomization.
The remainder is due to the fact that our $(L,f)$-replacement path covering consists of
$O\big(\frac{(L+f)^{L+f}}{L^L  f^f} \big) = O\big(\sqrt{f} \nwspace \binom{L+f}{f} \big)$ subgraphs (in expectation).
The lower bound, in turn, provides a graph with  $\Omega\big(\binom{L+f}{f}\big)$
pairwise conflicting failure sets.
Taken together, our upper and lower bounds essentially determine the covering value of $(L,f)$-replacement path coverings throughout the natural regime $f = O(L)$. The remaining gap is only polynomial in $f$ and logarithmic in $n$.
Another interesting question is to improve the query time
while retaining a \mbox{(near-)optimal} number of subgraphs.
\Cref{thm:upper} has the same drawback as \cite{WY13}
that finding the subfamily $\G_F$ amounts to scanning all of $\G$.
The more complicated techniques in
\cite{Bilo25IndexingSubnetworks,Bilo26SimplerRPCs,KarthikParter24DeterministicRPC_TALG}
enable a faster query but suffer from a higher covering value.

\begin{table}[t]
\caption{
Lower bounds on the covering value of $(L,f)$-replacement path coverings.
Symbol $\varepsilon$ stands for an arbitrarily small positive constant.
}
%\vspace*{.5em}
\centering
\setlength{\tabcolsep}{5pt}
\renewcommand{\arraystretch}{1.25}
\begin{tabular}{@{}ccc@{}}

\textbf{Covering Value} & \textbf{~~~~~~Range~~~~~~} & \textbf{Reference} \\
\noalign{\hrule height 1pt}\\[-10pt]

$\Omega\!\left( \left(\frac{L}{f}\right)^f \nwspace \right)$ &  $(L/f)^f \le n$ & \cite{KarthikParter24DeterministicRPC_TALG}\\[.75em]

$\Omega\!\left( \frac{\sqrt{f e^f}}{L} \cdot \left(\frac{L}{f}\right)^f \right)$ & $f \le (\frac{1}{2}{-}\varepsilon) \nwspace L$ & \cite{Bilo26SimplerRPCs}\\[.75em]

$\Omega\left(\frac{1}{\sqrt f}\cdot \frac{(L+f)^{L+f}}{L^L  f^f}\right)$ & $f = O(L)$ & \Cref{thm:lower}
\end{tabular}
\label{table:lower}
\end{table}

%\begin{figure}[!ht]
%\centering
%\vspace{2mm}
%\begin{tikzpicture}
%\begin{axis}[
%width=11cm, height=7cm,
%xlabel={$f$}, ylabel={$\ln(\text{covering value})$},
%domain=1:1000, samples=1000,
%legend pos=north west, legend cell align=left,
%grid=both, grid style={gray!20},
%every axis plot/.append style={very thick},
%]
%\pgfmathsetmacro{\L}{1000}
%
%\addplot[blue]
%{x*(ln(\L) - ln(x))};
%\addlegendentry{\cite{KarthikParter24DeterministicRPC_TALG}}
%
%\addplot[orange]
%{0.5*ln(x) + 0.5*x - ln(\L) + x*ln(\L) - x*ln(x)};
%\addlegendentry{\cite{Bilo26SimplerRPCs}}
%
%\addplot[red]
%{(\L+x)*ln(\L+x) - \L*ln(\L) - x*ln(x) - 0.5*ln(x)};
%\addlegendentry{New}
%
%\end{axis}
%\end{tikzpicture}
%\vspace{-3mm}
%\caption{Comparison of the three covering-value lower bounds on log scale (for $L=1000$). The new lower bound (in red) dominates prior work across the range $0<f\leq L$.}
%\label{figure:plot}
%\end{figure}

The paper is organized as follows:~\Cref{sec:upper} shows the upper bound proof of~\Cref{thm:upper};~\Cref{sec:lower} provides the lower bound construction of~\Cref{thm:lower}, and, finally, we compare the lower bound  with previous lower bounds in~\Cref{subsec:lower_comparison}.

\section{Upper Bound}
\label{sec:upper}

In this section, we briefly present the upper bound.
It makes explicit ideas that are already contained in \cite[Lemma~10]{Bilo22FixedParameterSensitivityOracles}
in the context of the $k$-path problem.
We provide a complete proof here for self-containment and readability.

\upperbound*

\begin{proof}
Let $G = (V,E)$ be the input graph and consider a spanning subgraph $G' = (V,E')$
that is obtained from $G$ by removing any edge in $E$ independently with probability $p$.
Let $s,t \in V$ be two vertices and $F \subseteq E$ a set of at most $f$ failing edges
such that there exists a shortest path $P$ from $s$ to $t$ in $G{-}F$ such that
$P$ has at most $L$ edges. Fix such a replacement path $P$.

We say $G'$ \emph{covers} $P$ if the subgraph contains none 
of the failing edges, i.e., $F \cap E' = \emptyset$, but all edges of the path, $E(P) \subseteq E'$.
This happens with probability $p^{|F|} (1-p)^{|E(P)|} \ge p^f (1-p)^L$.
The lower bound is maximum for $p = f/(f{+}L)$.
This results in a probability of at least
$q = p^f (1-p)^L = \frac{f^f L^L}{(f+L)^{f+L}}$
for any fixed hop-short replacement path to be covered.

By standard Chernoff bounds, see e.g.~\cite[p.~70]{Motwani95RandomizedAlgorithms},
sampling $O( f \log(n)/q)$ independent subgraphs $G'$ ensures
that a path is covered by at least one subgraph with probability $1-n^{\Omega(f)}$.
The number of relevant replacement paths is at most the number of triples
$(s,t,F)$, namely $|V|^2 \cdot \binom{|E|}{\le f} = O(n^2 \nwspace m^f) = O(n^{2f+2})$.
A union bound over all replacement paths shows that they are covered
with high probability.
The number of subgraphs is $O\left(f \log(n) \cdot \frac{(L+f)^{L+f}}{f^f L^L}\right)$.
\end{proof}

\section{Lower Bound}
\label{sec:lower}

In this section we show the lower bound stated in~\Cref{thm:lower}.
Throughout the remainder of the paper, assume that the cut-off value $L$ is at least 2.
The main part of the lower bound is to prove that any \mbox{$(L,f)$-replacement} path covering
requires at least $\binom{L+f-1}{f}$ subgraphs. 
We construct a directed graph $G$ with $L+1$ layers.
Layer $0$ only contains the source vertex $s$, and layer $L$ only contains the target $t$.
Each internal layer $i \in [L-1]$ has $f+1$ vertices $a_{i,0}, a_{i,1}, \dots, a_{i,f}$.
The assumption $L\ge 2$ ensures that there is at least one inner layer.
Edges exist only between neighboring layers.
\begin{itemize}
\item $s$ has an edge to every vertex of layer $1$;
\item every vertex $a_{L-1,j}$ in layer $L-1$ has an edge to $t$;
\item for $i \in [L-2]$, vertex $a_{i,j}$ has an edge to $a_{i+1,k}$ whenever $j \le k$.
\end{itemize}
Consider vertex $a_{i,j}$ in the $i$-th layer at level $j$.
We define the weight of each of its in-edges as $j \cdot (f+1)^{L-1-i}$.
The weight of the edges $(a_{{L-1},j},t)$ is irrelevant,
we set it to $0$ for simplicity.

\begin{figure}[!ht]
\centering
\includegraphics[width=0.58\linewidth]{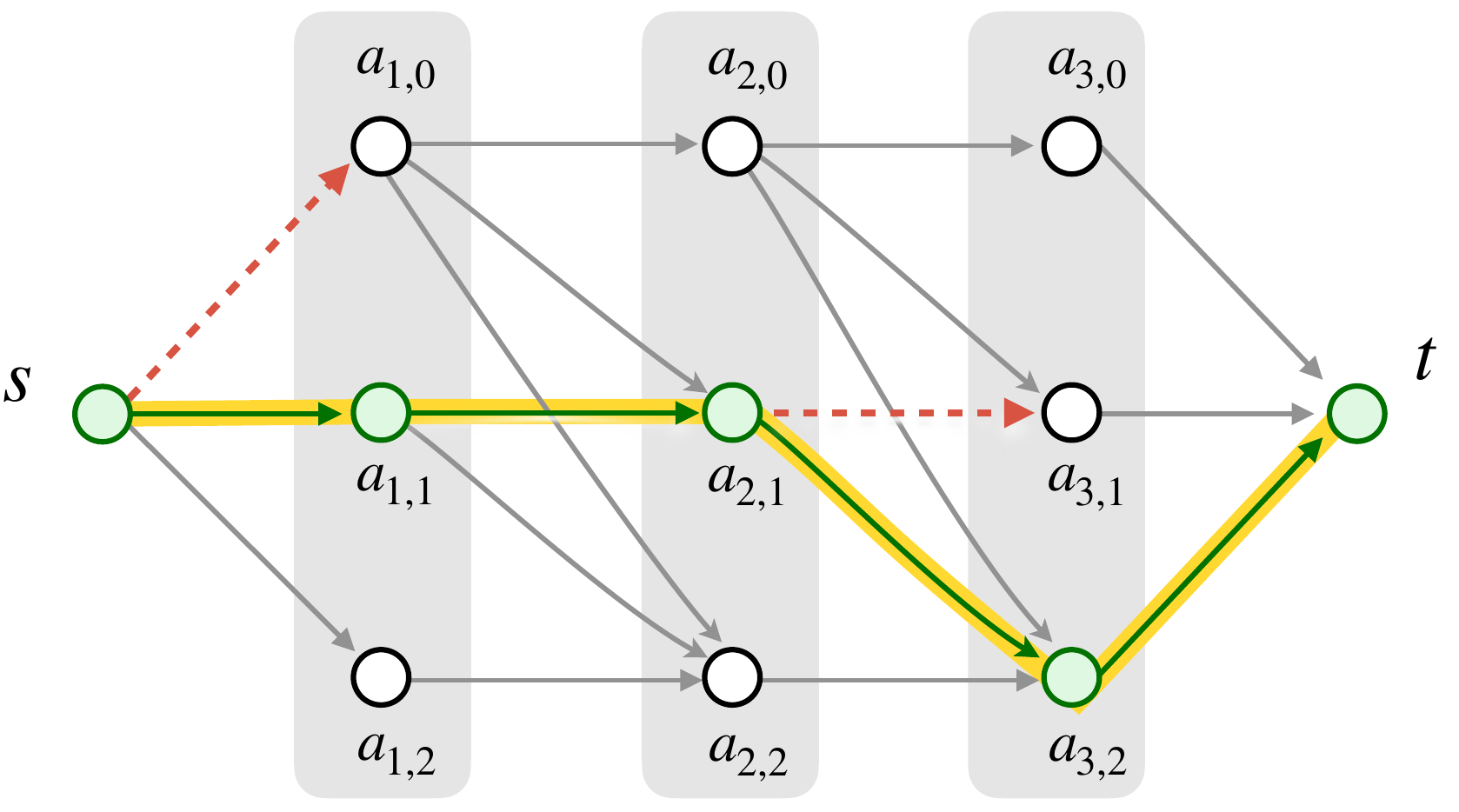}
\caption{The lower bound construction for paths of length $L=4$ and $f=2$ failures.}
\label{figure:lb}
\end{figure}

Note that every $s$-$t$ path in the graph $G$ uses exactly one vertex per layer, and thus consists of $L$ edges.
See \Cref{figure:lb} for an example.
Let $\mathcal{P}$ denote the set of all $s$-$t$ paths.
Any such path $P = (s, a_{1,j_1}, a_{2,j_2} \dots  a_{L-1,j_{L-1}},t) \in \mathcal{P}$ is uniquely defined by its non-decreasing level sequence $\ell(P) = (j_1, j_2, \dots, j_{L-1})$
where we have $j_i \in \{0, 1, \dots, f\}$ and $j_{i} \le j_{i+1}$ for each $i$.
When interpreting the entries of $\ell(P)$ as the digits of a base-$(f{+}1)$ integer,
the value of this integer is precisely the total weight of the path
%\begin{equation*}
	$\operatorname{wt}(P) = \sum_{i=1}^{L-1} j_i \cdot (f{+}1)^{L-1-i}$.
%\end{equation*}

\subsection{Covering $s$-$t$-Replacement Paths}
\label{subsec:lower_covering}

Let $\prec$ denote the (strict) lexicographic order on $\mathcal{P}$ based on these level sequences.
In more detail, for any two distinct paths $P \neq Q$ with level sequences
$\ell(P) = (j_1, \dots, j_{L-1})$ and $\ell(Q) = (k_1, \dots, k_{L-1})$,
let $i^*$ be the earliest layer in which $P$ and $Q$ take different vertices.
Then, $P \prec Q$ holds if and only if $j_{i^*} < k_{i^*}$.

\begin{lemma}
\label{lemma:weight-order}
	The total weight of a path is strictly increasing with respect to the lexicographical order. 
	That means, for distinct paths $P \neq Q \in \mathcal{P}$,
	$P \prec Q$ implies $\operatorname{wt}(P) < \operatorname{wt}(Q)$.
\end{lemma}

\begin{proof}
It suffices to prove that $\operatorname{wt}(Q) - \operatorname{wt}(P) \ge 1$.
Let $i^*$ be the first layer where the paths $P$ and $Q$ differ.
We thus have $j_i = k_i$ for all $i < i^*$, and $j_{i^*} < k_{i^*}$.
The difference in their total weights is
\begin{align*}
 	\operatorname{wt}(Q) - \operatorname{wt}(P) &= \sum_{i=1}^{L-1} (k_i - j_i) \cdot (f+1)^{L-1-i}
 		 = \sum_{i=i^*}^{L-1} (k_i - j_i) \cdot (f+1)^{L-1-i}\\
 		&= (k_{i^*} - j_{i^*}) \cdot (f+1)^{L-1-{i^*}} + \sum_{i=i^*+1}^{L-1} (k_i - j_i) \cdot (f+1)^{L-1-i}
\end{align*}

Since $j_{i^*}, k_{i^*}$ are integers with $j_{i^*} < k_{i^*}$, the leading term satisfies 
$(k_{i^*} - j_{i^*}) (f+1)^{L-1-{i^*}} \ge (f+1)^{L-1-{i^*}}$.
For every $i \ge i^*+1$ both levels lie in $\{0, 1, \dots, f\}$, 
hence $k_i - j_i \ge -f$.
Summing these bounds gives
\begin{align*}
\operatorname{wt}(Q) - \operatorname{wt}(P) &\ge (f+1)^{L-1-{i^*}} - \sum_{i=i^*+1}^{L-1} f \cdot (f+1)^{L-1-i} \\
	&= (f+1)^{L-1-{i^*}} - f \cdot \sum_{w=0}^{L-2-{i^*}} (f+1)^{w}\\
	&= (f+1)^{L-1-{i^*}} - \left( (f+1)^{L-1-{i^*}} - 1 \right)
	 = 1.
\end{align*}
\end{proof}
 
%\paragraph{The failure set of a path.}
Fix an $s$-$t$ path $P$ with level sequence $\ell(P) = (j_1,\dots,j_{L-1})$. 
We define a set of failing edges $F_P$ associated with this path. Intuitively, $F_P$ removes exactly those edges that would allow a competing path to become lexicographically smaller.
The set contains, at each vertex visited by $P$, 
every edge that goes to a \emph{lower} level than the one $P$ takes next.

\begin{equation*}
	F_P = \left\{ (s, a_{1,k}) \mid 0 \le k < j_1 \right\} \cup \left\{ (a_{i,  j_{i}} , a_{i+1,  k})  \mid i \in [L-2],\ \ j_{i} \le k < j_{i+1}   \right\}.
\end{equation*}
None of these edges lie on the path because $P$ always moves to level $j_{i+1}$,
not to anything below it. 
Also, by construction of the graph $G$,
edges between layers $i$ and $i+1$ only connect vertices in the same or higher levels.
Therefore, if $P$ does not change levels ($j_i = j_{i+1}$),
no edges between those layers are deleted by $F_P$.
At layer $i$, the failure set contains exactly $j_{i+1} - j_{i}$ of the out-edges.
Hence, $|F_P| = j_1 + \sum_{i=1}^{L-2} (j_{i+1} - j_{i}) = j_{L-1}   \le   f$
satisfies the sensitivity bound.

\begin{lemma}
\label{lemma:unique}
For every $P \in \mathcal{P}$, the path $P$ is the unique shortest $s$-$t$ path in $G{-}F_P$.
\end{lemma}

\begin{proof}
Recall that every $s$-$t$ path has $L$ edges
and that the path $P$ survives in the subgraph $G{-}F_P$. 
We show that it has minimum total weight among all $s$-$t$ paths in $G{-}F_P$.
Let $Q \in \mathcal{P}{\setminus}\{P\}$ be any other path with $E(Q) \cap F_P = \emptyset$.
Let $\ell(P) = (j_1, \dots, j_{L-1})$ and $\ell(Q) = (k_1,\dots,k_{L-1})$
be the respective level sequences and
$i^*$ the first layer in which $Q$ and $P$ differ.
First, assume $i^* > 1$.
The path $Q$ thus arrives at the vertex $a_{i^*-1,j_{i^*-1}}$.
The levels $[0, j_{i^*-1})$ are not reachable from $a_{i^*-1,j_{i^*-1}}$ by the construction of $G$
and the failure set $F_P$ contains every out-edge to a level in $[j_{i^*-1}, j_{i^*})$,
hence $j_{i^*} \le k_{i^*}$.
If $i^* = 1$, then the same conclusion holds simply from $F_P$
containing all edges $(s, a_{1,k})$ for $k < j_1$.
Since $P$ and $Q$ differ in layer $i^*$, we have $j_{i^*} < k_{i^*}$ and thus $P \prec Q$.
\Cref{lemma:weight-order} now implies $\operatorname{wt}(P) < \operatorname{wt}(Q)$.
\end{proof}
 
\begin{lemma}
\label{lemma:count}
	An $(L,f)$-replacement path covering for the graph $G$ must cover 
	different $s$-$t$ paths with different subgraphs.
	Therefore, any RPC has covering value at least $|\mathcal{P}| = \binom{L+f-1}{f}$.
\end{lemma}
 
\begin{proof}
The number of $s$-$t$ paths in $\mathcal{P}$ is the same as the number of non-decreasing level sequences $0 \le j_1 \le \dots \le j_{L-1} \le f$.
A straightforward stars-and-bars argument with $L-1$ stars and $f+1$ bars shows
that this number is $\binom{(L-1)+(f+1)-1}{(f+1)-1} = \binom{L+f-1}{f}$.

Let $\G$ be an $(L,f)$-replacement path covering for $G$.
Fix a path $P \in \mathcal{P}$.
It is the $s$-$t$-replacement path for the failure set $F_P$ and has $L$ edges.
Recall from \Cref{def:RPC} that $\G_{F_P} \subseteq \G$ is the subfamily
that is relevant for any query involving $F_P$ (in particular, for $(s,t,F_P)$).
There exists a subgraph $\sigma(P) \in \G_{F_P}$
that contains all edges of $P$ and none of $F_P$. 
This defines a map $\sigma \colon \mathcal{P} \to \G$.
We show that $\sigma$ is injective.
 
Let $Q \in \mathcal{P}{\setminus}\{P\}$ be another $s$-$t$ path.
To reach a contradiction, assume $P$ and $Q$ are covered by the same subgraph in the RPC,
i.e., $\sigma(P) = \sigma(Q)$.
As before, let $i^*$ denote the first layer in which the paths differ.
We can assume $P \prec Q$ without loss of generality, that means, $j_{i^*} < k_{i^*}$. 
The two paths are the same up until vertex $a_{i^*-1,j_{i^*-1}}$ 
(reps.\ the vertex $s$ if $i^* = 1$).
Path $P$ uses the edge $(a_{i^*-1,j_{i^*-1}}, a_{i^*,j_{i^*}})$ 
(resp.\ the edge $(s, a_{1,j_1})$).
As the level $j_{i^*}$ is below $k_{i^*}$,
the edge belongs to the failure set $F_Q$ of the other path.
By the properties of an $(L,f)$-replacement path covering,
$(a_{i^*-1,j_{i^*-1}}, a_{i^*,j_{i^*}})$ is not contained in the edge set $E(\sigma(Q))$. 
On the other hand, the same edge lies on the path $P$ and is therefore in $E(\sigma(P))$,
a contradiction to $\sigma(P) = \sigma(Q)$.
\end{proof}

We now estimate the binomial coefficient via Stirling's approximation.

\begin{lemma}
\label{lemma:lower_directed}
For positive integers $f = O(L)$, it holds that
\begin{equation*}
	\binom{L+f-1}{f} = \Omega\!\left(\frac{1}{\sqrt f}\cdot \frac{(L+f)^{L+f}}{f^f L^L}\right).
\end{equation*} 
\end{lemma}
 
\begin{proof}
It is more convenient to work with $\binom{L+f}{f}$ instead.
This does not change the asymptotics since $\binom{L+f-1}{f} = \frac{L}{L+f} \binom{L+f}{f} = \Omega\big(\binom{L+f}{f} \big)$
The last estimate is due to $f = O(L)$.

Stirling's formula (see e.g.~\cite{Robbins55Stirling}) states
\begin{equation*}
	\binom{N}{k} = \Theta\!\left(\sqrt{\frac{N}{k(N-k)}}\cdot \frac{N^N}{k^k \nwspace (N-k)^{N-k}}\right).
\end{equation*}
Substituting $N = L+f$ and $k = f$ (and thus $N-k = L$) simplifies the pre-factor to
$\sqrt{\frac{L+f}{fL}}$.
The assumption $f = O(L)$ implies that $(L{+}f)/L$ is a constant,
hence the factor is of order $\Theta(1/\sqrt{f} \nwspace)$.
The other term evaluates to $(L{+}f)^{L+f}/(f^f L^L)$.
\end{proof}

This completes the proof of \Cref{thm:lower}.

\subsection{Comparison with Previous Lower Bounds}
\label{subsec:lower_comparison}

We compare our bound with the prior results
$\Omega((L/f)^f\big)$ \cite{KarthikParter24DeterministicRPC_TALG}
and $\Omega((\sqrt{f e^f}/L) \cdot (L/f)^f)$ \cite{Bilo26SimplerRPCs}.
They both share the factor $M = (L/f)^f$, so we first express ours in the same form.

\begin{lemma}
For positive integers $f, L$, it holds that
%\begin{equation*}
	$\frac{(L+f)^{L+f}}{L^L f^f} \ge e^{f} \cdot M$.
%\end{equation*}
\end{lemma}

\begin{proof}
Factoring
$\frac{(L+f)^{L+f}}{L^L f^f} = \left(1+\tfrac fL\right)^{L+f}\left(\tfrac Lf\right)^{f}$
reduces the claim to $\left(1+\tfrac{f}{L}\right)^{L+f} \ge e^f$.
Taking the $L$-th root gives $\left(1+\tfrac{f}{L}\right)^{\frac{L+f}{L}} \ge e^{\frac{f}{L}}$,
and by taking the logarithm we arrive at
\begin{equation*}
	\left(1+\tfrac{f}{L}\right) \cdot \ln\!\left(1+\tfrac{f}{L} \right) \ge \frac{f}{L}.
\end{equation*}
To see this, define $t = \frac{f}{L}$ and
$h(t) = (1+t) \ln(1+t) - t$.
We need to show that $h(t)$ is non-negative for $t \ge 0$.
This follows from $h(0) = 0$ and $\frac{\mathrm{d}}{\mathrm{d}t} h(t) = \ln(1+t) \ge 0$. 
\end{proof}

Our bound in \Cref{thm:lower} can thus be written as 
$\Omega( (e^f/\sqrt{f}) \cdot M)$.
Comparing this to $\Omega(M)$ in \cite{KarthikParter24DeterministicRPC_TALG} and
$\Omega((\sqrt{f e^f}/L) \cdot M)$ in \cite{Bilo26SimplerRPCs}. 
The gain over the latter bound is $\Omega(\sqrt{e^f}L/f)$,
which is exponential in $f$ for all $f = O(L)$. 
In particular, our lower bound uniformly subsumes all previously known lower bounds for $f = O(L)$, while remaining within a factor of $O(\sqrt{f})$ of the corresponding upper bound (ignoring logarithmic factors).

\section*{Acknowledgments}

The first author is supported by the project ``SOS-TG: Spanners and Oracles for Static and Temporal Graphs'', funded by Università degli Studi dell'Aquila under the Call for Proposals for Fundamental Research and Early-Career Research Grants -- Year 2026. The second author is supported by the Indian Anusandhan National Research Foundation (ANRF) under the Mathematical Research Impact-Centric Support (MATRICS) scheme, grant agreement No. MTR/2025/001601.
The fourth author is supported by the German Research Foundation (DFG),
grant agreement No.~556899211 ``Design, Analysis, and Engineering of Enumeration Algorithms''.

\bibliographystyle{alphaurl}
\bibliography{ref}

\end{document}